\documentclass{article}

\usepackage[square,numbers]{natbib}

\usepackage{hyperref}
\usepackage{a4wide}
\usepackage{amsmath}
\usepackage{amssymb}
\usepackage{amsthm}
\usepackage{dsfont}
\usepackage{graphicx}
\usepackage{comment}
\usepackage[utf8]{inputenc}
\usepackage{tikz}
\usetikzlibrary{calc, positioning, fit, shapes.misc}

\newtheorem{theorem}{Theorem}
\newtheorem{lemma}[theorem]{Lemma}
\newtheorem{corollary}[theorem]{Corollary}
\newtheorem{claim}[theorem]{Claim}
\newtheorem{proposition}[theorem]{Proposition}
\newtheorem{experiment}[theorem]{Random Experiment}
\newtheorem{remark}[theorem]{Remark}

\newtheorem{innercustomthm}{Lemma}
\newenvironment{customthm}[1]
  {\renewcommand\theinnercustomthm{#1}\innercustomthm}
  {\endinnercustomthm}

\newtheorem{innercustomcla}{Claim}
\newenvironment{customcla}[1]
  {\renewcommand\theinnercustomcla{#1}\innercustomcla}
  {\endinnercustomcla}

\newcommand*{\email}[1]{%
    \href{mailto:#1}{#1}\par
    }

\newcommand{\OPT}{\mathrm{OPT}}
\newcommand{\E}{\mathbb{E}}
\renewcommand{\P}{\mathbb{P}}
\newcommand{\C}{\mathcal{C}}
\newcommand{\K}{\mathcal{K}}
\newcommand{\I}{R}
\newcommand{\F}{\mathcal{F}}
\newcommand{\N}{\mathcal{N}}

\usepackage[colorinlistoftodos]{todonotes}

\title{The Combinatorial Santa Claus Problem or:\\ How to Find Good Matchings in Non-Uniform Hypergraphs\footnote{This research was supported by the Swiss National Science Foundation project
200021-184656 “Randomness in Problem Instances and Randomized Algorithms.”}}

\author{Etienne Bamas\footnote{EPFL, Switzerland, \email{etienne.bamas@epfl.ch}} \and Paritosh Garg\footnote{EPFL, Switzerland, \email{paritosh.garg@epfl.ch}} \and Lars Rohwedder\footnote{EPFL, Switzerland, \email{lars.rohwedder@epfl.ch}}}

\date{}

\begin{document}

\maketitle

\begin{abstract}
    We consider hypergraphs on vertices $P\cup R$ where each hyperedge contains exactly one vertex in $P$. Our goal is to select a matching
    that covers all of $P$, but we allow each selected
    hyperedge to drop all but an $(1/\alpha)$-fraction of its intersection with $R$ (thus relaxing the matching constraint). Here $\alpha$ is to be minimized.
    We dub this problem the Combinatorial Santa Claus problem, since we show in this paper that this problem and the Santa Claus problem are almost equivalent in terms of their approximability. 
    
    The non-trivial observation that any uniform regular hypergraph
    admits a relaxed matching for $\alpha = O(1)$
    was a major step in obtaining a constant approximation rate for a special case of the
    Santa Claus problem, which received great attention in literature.
    It is natural to ask if the
    uniformity condition can be omitted.
    Our main result is that every (non-uniform) regular hypergraph admits a relaxed matching for $\alpha = O(\log\log(|R|))$,
    when all hyperedges are sufficiently large (a condition that is necessary).
    In particular, this implies an $O(\log\log(|R|))$-approximation algorithm
    for the Combinatorial Santa Claus problem with large hyperedges.
\end{abstract}
\pagebreak

\section{Introduction}
Let $\mathcal H = (P\cup R, \mathcal C)$ be a hypergraph with
hyperedges $\mathcal C$ over the vertices $P$ and $R$ with
$m = |P|$ and $n = |R|$.
Each hyperedge $C\in\mathcal C$ contains exactly one vertex in $P$,
that is, $|C\cap P| = 1$.
The hypergraph may contain multiple copies of the same hyperedge.
We are looking for a matching that covers all vertices in $P$.
As one of Karp's $21$ NP-complete problems~\cite{DBLP:conf/coco/Karp72},
this problem is already NP-hard for $3$-uniform hypergraphs.
Hence, we consider a relaxation of the following kind.
Each selected hyperedge is allowed to drop all but a $(1/\alpha)$-fraction of its
intersection with $R$.
More formally, we seek to find a set $\mathcal K$ with
$P\subseteq \bigcup \mathcal K$
where for each $K\in\mathcal K$ there is
an edge $C\in\mathcal C$ with $K\subseteq C$
and $|K|\ge (1/\alpha) |C|$. We call such a solution an $\alpha$-relaxed perfect matching.
One seeks to minimize $\alpha$.

This problem is a special case of the Santa Claus
problem, in which we distribute resources to players,
each player $i$ has a valuation $v_{ij}$ for each resource $j$, and we want to maximize the value for the least happy player
(the sum of his valuations over resources assigned to him).
The best approximation algorithm known for this problem has an approximation rate of
$n^{\epsilon}$ and runs in time $n^{O(1/\epsilon)}$~\cite{DBLP:conf/focs/ChakrabartyCK09} for $\epsilon\in\Omega(\log\log(n)/\log(n))$.
On the negative side,
it is only known
that computing a $(2 - \delta)$-approximation is NP-hard.
We elaborate later a straight-forward reduction of the hypergraph problem
to the Santa Claus problem preserving the approximation rate and a more involved
reduction that turns a $c$-approximation rate
for the hypergraph problem into a
$O((c\log^*(n))^2)$-approximation rate for the 
Santa Claus problem.
In view of this connection, we dub the hypergraph problem the \emph{Combinatorial Santa Claus} problem and
refer to $P$ as the players, $R$ as the resources, and $\mathcal C$
as the configurations.

This problem was studied previously for uniform
(each edge has the same size) and regular
(each vertex in $P\cup R$ has the same degree) hypergraphs.
\citet{BansalSrividenko} have shown that any hypergraph with both properties admits
an $\alpha$-relaxed perfect matching for $\alpha = O(\log\log(n))$\footnote{In fact, they get $\alpha = O(\log\log(m) / \log\log\log(m))$ by a slightly more careful analysis.}.
This was also motivated by the Santa Claus problem.
Namely, it implies an $O(\alpha)$ integrality
gap for an LP relaxation of the restricted Santa Claus problem, which is the special case where
for all valuations $v_{ij}\in\{0,v_j\}$
($v_j$ being a player independent value).
This was subsequently improved by \citet{Feige} to $\alpha = O(1)$. Further improvements on the constant were
obtained since then (see end of the section).

It is natural to ask, whether regularity is
also a sufficient condition for the existence of
an $\alpha$-relaxed perfect matching
in the case of non-uniform hypergraphs
(for some small $\alpha$).
The simple answer is that this is not the case.
Figure~\ref{fig:counter-example} (derived from a similar construction in~\cite{BansalSrividenko}) shows that there
are regular hypergraphs that do not contain an $\alpha$-relaxed perfect matching for $\alpha \approx \sqrt{n}$.
However, a crucial aspect of the counter-example
is that some hyperedges contain just one element of $R$.
Those hyperedges cannot drop any elements in an $\alpha$-relaxed matching.
Our main result is that for $\alpha = O(\log\log(n))$
a regular (non-uniform)
hypergraph where all hyperedges
contain at least $\alpha$ elements admits
a $\alpha$-relaxed perfect matching.
We note that the setting with the assumption that
hyperedges are large still contains the uniform case without the assumption
(and therefore is much more general):
This holds because the case where all hyperedges
are smaller than $\alpha$ trivially admits an
$\alpha$-relaxed perfect matching by Hall's condition.
In particular, our result directly implies a bound comparable to Bansal and Srividenko
in the uniform case.

Our result for regular hypergraphs implies that we can approximate the Combinatorial Santa Claus problem on
non-regular, non-uniform hypergraphs, when all edges are sufficiently large.
More precisely, we can either determine that there is no $\delta$-relaxed perfect matching
or compute a $(\delta / \alpha)$-relaxed perfect matching for $\alpha = O(\log\log(n))$,
when all hyperedges are of size at least $\alpha / \delta$.

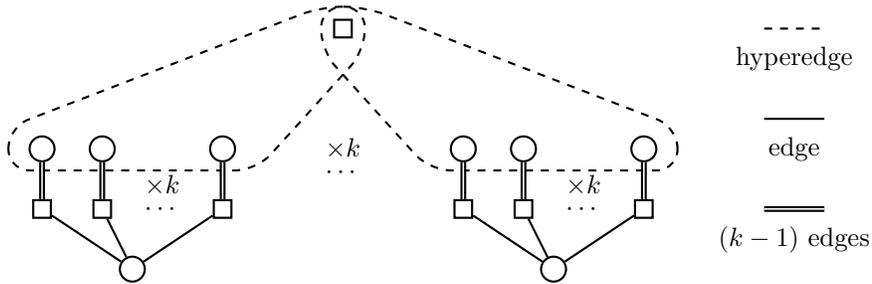
\begin{figure}
    \centering
    \begin{tikzpicture}[scale=0.8]
    \node[draw, thick, rectangle] (p0) at (5, 5) {};
    
    \node[draw, thick, rectangle] (p1) at (0, 2) {};
    \node[draw, thick, circle] (r1) at (0, 3) {};
    \node[draw, thick, rectangle] (p2) at (1, 2) {};
    \node[draw, thick, circle] (r2) at (1, 3) {};
    \node at (2, 2) {$\cdots$};
    \node at (2, 2.4) {$\times k$};
    \node[draw, thick, rectangle] (p3) at (3, 2) {};
    \node[draw, thick, circle] (r3) at (3, 3) {};
    \node[draw, thick, circle] (r4) at (1.5, 1) {};
    \node (h1) at (0.5, 1) {};
    \node (h2) at (2.5, 1) {};
    
    \draw[thick, dashed, rounded corners=8pt]
    ($(r1.south west)+(-0.4,-0.2)$) -- 
    ($(r1.north west)+(-0.4,0.2)$) -- 
    ($(p0.north west)+(-0.2,0.2)$) -- 
    ($(p0.north east)+(0.2,0.2)$) -- 
    ($(p0.south east)+(0.2,-0.2)$) -- 
    ($(r3.south east)+(0.4,-0.2)$) --cycle;
    
    \draw[thick, double] (r1) -- (p1);
    \draw[thick, double] (r2) -- (p2);
    \draw[thick, double] (r3) -- (p3);
    \draw[thick] (r4) -- (p1);
    \draw[thick] (r4) -- (p2);
    \draw[thick] (r4) -- (p3);
    
    \node at (5, 3) {$\times k$};
    \node at (5, 2.6) {$\cdots$};
    
    \node[draw, thick, rectangle] (p1') at (7, 2) {};
    \node[draw, thick, circle] (r1') at (7, 3) {};
    \node[draw, thick, rectangle] (p2') at (8, 2) {};
    \node[draw, thick, circle] (r2') at (8, 3) {};
    \node at (9, 2) {$\cdots$};
    \node at (9, 2.4) {$\times k$};
    \node[draw, thick, rectangle] (p3') at (10, 2) {};
    \node[draw, thick, circle] (r3') at (10, 3) {};
    \node[draw, thick, circle] (r4') at (8.5, 1) {};
    \node (h1') at (7.5, 1) {};
    \node (h2') at (9.5, 1) {};

    \draw[thick, dashed, rounded corners=8pt]
    ($(r3'.south east)+(0.4,-0.2)$) -- 
    ($(r3'.north east)+(0.4,0.2)$) -- 
    ($(p0.north east)+(0.2,0.2)$) -- 
    ($(p0.north west)+(-0.2,0.2)$) -- 
    ($(p0.south west)+(-0.2,-0.2)$) -- 
    ($(r1'.south west)+(-0.4,-0.2)$) --cycle;
    
    \draw[thick, double] (r1') -- (p1');
    \draw[thick, double] (r2') -- (p2');
    \draw[thick, double] (r3') -- (p3');
    \draw[thick] (r4') -- (p1');
    \draw[thick] (r4') -- (p2');
    \draw[thick] (r4') -- (p3');
    
    \draw[thick, dashed] (12, 5) -- (13, 5);
    \node at (12.5, 4.5) {hyperedge};
    
    \draw[thick] (12, 3.5) -- (13, 3.5);
    \node at (12.5, 3) {edge};
    
    \draw[thick, double] (12, 2) -- (13, 2);
    \node at (12.5, 1.5) {$(k - 1)$ edges};
    
    \end{tikzpicture}
    \caption{A $k$-regular hypergraph with $n = k(k + 1)$
    and no $\alpha$-relaxed perfect matching for any $\alpha < k$.}
    \label{fig:counter-example}
\end{figure}

\paragraph*{Overview of techniques.}
We note that for $\alpha = O(\log(n))$ the statement is
easy to prove: We select for each player $i$
one of the configurations containing $i$ uniformly
at random. Then by standard concentration bounds each
resource is contained in at most $O(\log(n))$
of the selected configurations with high probability.
This implies that there is a fractional assignment of resources to
configurations, such that each of the
selected configurations $C$ receives
$\lfloor |C| / O(\log(n)) \rfloor$ of the resources
in $C$. By integrality of the bipartite matching polytope, there
is also an integral assignment with this property.

We now briefly describe the approach of Bansal and
Srividenko to improve this to $O(\log\log(n))$ in
the uniform case.
Let $k$ be the size of each configuration.
First they reduce the degree
of each player and resource to $O(\log(n))$ using
the argument above, but taking $O(\log(n))$ configurations for each player.
Then they observe that when sampling
uniformly at random $O(n \log(n) / k)$ resources
and dropping all others, an
$\alpha$-relaxed perfect matching with respect to
the smaller set of resources is still an
$O(\alpha)$-relaxed perfect matching with respect
to all resources with high probability (when assigning the dropped resources
to the selected configurations appropriately).
Indeed, the smaller instance is easier to solve:
With high probability all configurations have size $O(\log(n))$ and this greatly reduces the dependencies
between the bad events of the random experiment above
(the event that a resource is contained
in too many selected configurations).
This allows them to apply Lov\'asz Local Lemma (LLL)
in order to show that with positive probability the
experiment succeeds for $\alpha = O(\log\log(n))$.
Feige's improvement uses a more involved variant of
these arguments allowing him to iterate the approach
and repeatedly reduce the degree of of the players
and the size of the configurations.

It is not obvious how to extend this approach
to non-uniform hypergraphs:
Sampling a fixed fraction of the
resources will either make the small configurations
empty---which makes it impossible to retain guarantees for the original instance---or it leaves the big configurations big%
---which fails to reduce the dependencies enough to
apply LLL. Hence it requires new sophisticated ideas
for non-uniform hypergraphs, which we describe next.

Suppose we are able to find a set $\mathcal K\subseteq \mathcal C$ of configurations (one for each player)
such that for each $K\in\mathcal K$ the sum of intersections $|K\cap K'|$ with smaller configurations $K'\in \mathcal K$ is
very small, say at most $|K| / 2$.
Then it is easy to derive a $2$-relaxed perfect matching: We iterate over all $K\in\mathcal K$ from large to small and reassign all resources to $K$
(possibly stealing them from the configuration that previously had them).
In this process every configuration gets stolen at most $|K| / 2$ of its resources,
in particular, it keeps the other half.
Indeed, it is non-trivial to obtain a property like the one mentioned above. If we take a random configuration for each player, the
dependencies of the intersections are too complex.
To avoid this we invoke an advanced variant of the sampling approach where we construct not only one
set of resources, but a hierarchy of
resource sets $\I_0\supseteq \cdots \supseteq \I_d$ by repeatedly dropping a fraction of resources from the previous set.
We then formulate bad events based on the intersections
of a configuration $C$ with smaller configurations $C'$,
but we write it only considering a resource set $\I_k$ of convenient granularity
(chosen based on the size of $C'$).
In this way we formulate a number of bad events using
various sets $\I_k$. This succeeds in reducing
the dependencies enough to apply LLL.
Then we derive from the properties we obtained
an $\alpha$-relaxed perfect matching, which is
again significantly more involved than the corresponding
argument in the uniform case.

\paragraph*{Other related work.}
For the uniform case the best constant achieved is due to a local search
procedure~\cite{DBLP:journals/talg/AsadpourFS12, DBLP:journals/talg/AnnamalaiKS17, DBLP:conf/soda/DaviesRZ20, DBLP:conf/icalp/ChengM19} that was discovered after the LLL method by~\cite{BansalSrividenko,Feige}.
This procedure is a generalization of the classical Hungarian method for maximum
matching in a bipartite graph~\cite{kuhn1955hungarian}.
It repeatedly augments a given matching
swapping hyperedges of a current matching
for others, eventually increasing the size of the matching.
The approach appears to be less suited for the non-uniform case.
This is because its analysis (which is closely coupled to the design of the algorithm)
heavily relies on amortizing resources of hyperedges
with the resources of other hyperedges and then applying volume arguments.
In the non-uniform case the cardinalities of the resources of two hyperedges
are not related, which makes it hopeless to amortize them.

\section{Matchings in regular hypergraphs}
In this section we give a proof of our
main result, namely that every regular hypergraph with hyperedges of size at least $\alpha$ has
an $\alpha$-relaxed perfect matching for
$\alpha = O(\log\log(n))$.
We start by defining our notation.
Let $\ell \in \mathbb N$.
We assume that each player $i\in P$ has $\ell$ configurations $\C_i$ which he intersects with.
Moreover, each resource $j\in R$ appears in at most $\ell$ configurations.
Note that this is a slightly more general condition than
regularity (requiring that each vertex
has the same degree of exactly $\ell$).
We can assume like hinted in the introduction that $\ell =300.000\log^{3}(n)$ at a constant loss
(see Appendix \ref{appendix_main} for details).

We now group the configurations in $\C_i$ by size:
We denote by $\C_i^{(0)}$ the configurations
of size in $[1,\ell^{4})$ and
for $k\ge 1$ we write $\C_i^{(k)}$ for the configurations
of size in $[\ell^{k+3},\ell^{k+4})$.
Moreover, define
$\C^{(k)}=\bigcup_i \C_i^{(k)}$
and $\C^{(\ge k)} = \bigcup_{h\ge k} \C^{(h)}$.
Let $d$ be the
smallest number such that
$\C^{(\ge d)}$ is empty. Note that
$d\le \log(n) / \log(\ell)$.
Now consider the following random process.
\begin{experiment}\label{exp:sequence}
We construct a nested sequence of resource sets $\I=\I_0 \supseteq \I_1 \supseteq  \ldots \supseteq  \I_d$ as follows. Each $\I_k$ is obtained from $\I_{k-1}$ by deleting every resource in $\I_{k-1}$ independently with probability $(\ell-1) / \ell$.
\end{experiment}
In expectation only a $1/\ell$ fraction of resources in $\I_{k-1}$ survives in $\I_k$.
Also notice that for $C \in \C^{(k)}$ we
have that $\mathbb E[ |\I_k \cap C| ] = \mathrm{poly}(\ell)$.

We proceed as follows.
In Section~\ref{sec:sequence}, we give some properties
of the resource sets constructed by
Random Experiment~\ref{exp:sequence} that hold
with high probability.
Then in Section~\ref{sec:LLL}, we show that we can
find a single configuration for each player such
that the intersection with smaller selected configurations
is bounded if we restrict the resource set to an appropriate
$\I_k$. Restricting the resource set is important to bound
the dependencies of bad events in order to apply Lovasz
Local Lemma.
Finally in Section~\ref{sec:reconstruction},
we demonstrate that these configurations also give
an $\alpha$-relaxed perfect matching for an
appropriate assignment of resources to configurations.

\subsection{Properties of resource sets}\label{sec:sequence}
In this subsection, we give a precise statement of the key properties that we need from Random Experiment~\ref{exp:sequence}. The first two lemmas have a straight-forward proof. The last one is a generalization of an argument used by \citet{BansalSrividenko}. Since the proof is more technical and tedious, we also defer it to Appendix~\ref{appendix_sequence} along with the proof of the first two statements.

We start with the first property which bounds the size of the 
configurations when restricted to some $\I_k$.
\begin{lemma}
\label{lma-size}
Consider Random Experiment~\ref{exp:sequence}
with $\ell\geq 300.000\log^{3} (n)$.
For any $k\geq 0$ and any $C\in\C^{(\geq k)}$ we have 
  \begin{equation*}
      \frac{1}{2} \ell^{-k}|C| \le |\I_k \cap C| \le \frac{3}{2} \ell^{-k}|C|
  \end{equation*}
with probability at least $1-1/n^{10}$.
\end{lemma}
The next property expresses that for any configuration the sum of intersections with configurations of a particular size does not deviate much from its expectation. 
\begin{lemma}
\label{lma-overlap-representative}
Consider Random Experiment~\ref{exp:sequence}
with $\ell\geq 300.000\log^{3} (n)$.
For any $k\geq 0$ and any $C\in\C^{(\geq k)}$ we have 
  \begin{equation*}
      \sum_{C'\in \C^{(k)}} |C'\cap C\cap \I_k| \leq \frac{10}{\ell^{k}} \left(|C|+\sum_{C'\in \C^{(k)}} |C'\cap C| \right)
  \end{equation*}
with probability at least $1-1/n^{10}$.
\end{lemma}
We now define the notion of \emph{good} solutions which is helpful in stating our last property.
Let $\F$ be a set of configurations,
$\alpha:\F \rightarrow \mathbb N$, $\gamma \in\mathbb N$, and $\I'\subseteq \I$.
We say that an assignment of $\I'$ to $\F$ is $(\alpha,\gamma)$-good if every configuration $C\in \F$ receives at least $\alpha(C)$ resources of $C\cap \I'$ and if no resource in $\I'$ is assigned more than $\gamma$ times in total.

Below we obtain that given a $(\alpha,\gamma)$-good solution with respect to resource set $\I_{k+1}$, one can construct an almost $(\ell \cdot \alpha,\gamma)$-good solution with respect to the bigger resource set $\I_{k}$. 
\begin{lemma}
\label{lma-good-solution}
Consider Random Experiment~\ref{exp:sequence}
with $\ell\geq 300.000\log^{3} (n)$.
Assume that the bounds in Lemma~\ref{lma-size} hold
for some $k\ge 0$.
Then with probability at least $1 - 1/n^{10}$
the following holds for all
$\F\subseteq \C^{(\geq k+1)}$, $\alpha:\F \rightarrow \mathbb N$, and $\gamma \in\mathbb N$ such that $\ell^3/1000\leq \alpha(C) \leq n $ for all $C\in\F$ and
$\gamma \in \{1,\dotsc,\ell\}$:
If there is a $(\alpha,\gamma)$-good assignment of $\I_{k+1}$ to $\F$, then there is a $(\alpha',\gamma)$-good assignment of $\I_k$ to $\F$ where
\begin{equation*}
    \alpha'(C) \ge \ell \left(1-\frac{1}{\log (n)} \right) \alpha(C)
\end{equation*}
for all $C\in\F$.
Moreover, this assignment can be found in polynomial time.
\end{lemma}

Given the lemmata above, by a simple union bound one gets that all the properties of resource sets hold.

\subsection{Selection of configurations}\label{sec:LLL}
In this subsection, we give a random process that selects one configuration
for each player such that the intersection with smaller configurations is bounded when
considering on appropriate sets $R_k$.
For clarity, we state in the following lemma what the properties of the sets
$R_0,\dotsc,R_d$ that we need are. These hold with high probability by the lemmata of
the previous section.

\begin{lemma}\label{lma:main-LLL}
  Let $\I =\I_0\supseteq\dotsc\supseteq\I_d$ be
  sets of fewer and fewer resources.
  Assume that for each $k$ and $C\in \mathcal C_i^{(k)}$ 
  we have
  \begin{equation*}
      1/2 \cdot \ell^{k - h} \le |C\cap \I_h| \le 3/2 \cdot \ell^{- h} |C| < 3/2 \cdot \ell^{k - h + 4}
  \end{equation*}
  for all $h=0,\dotsc,k$.
  Then there are $\mathcal K_i^{(k)}\subseteq \mathcal C_i^{(k)}$ such that
  $|\mathcal K_i| = 1$ and for each $k=0,\dotsc,d$, $j=0,\dotsc,k$ and
  $C\in\mathcal C^{(k)}$ we have
  \begin{equation*}
      \sum_{j\leq h\le k} \sum_{K\in\mathcal K^{(h)}} \ell^{h} |K \cap C \cap \I_h|
      \le \frac{1}{\ell} \sum_{j\leq h\le k} \sum_{C'\in\mathcal C^{(h)}} \ell^{h} |C' \cap C \cap \I_h| + 1000 \frac{d + \ell}{\ell}\log(\ell) |C| .
  \end{equation*}
\end{lemma}
Before we prove this lemma, we give an intuition of the statement.
Consider the sets $\I_1,\dotsc,\I_d$ constructed as in Random Experiment~\ref{exp:sequence}.
Then for $C'\in\mathcal C^{(h)}$ we have $\E[\ell^h |C'\cap C\cap \I_h|] = |C'\cap C|$.
Hence
\begin{equation*}
  \sum_{h\le k} \sum_{K\in\mathcal K^{(h)}} |K \cap C| = \E[\sum_{h\le k} \sum_{K\in\mathcal K^{(h)}} \ell^h |K \cap C \cap \I_h|]
\end{equation*}
Similarly for the right-hand side we have
\begin{multline*}
  \E[\frac{1}{\ell} \sum_{j \le h\le k} \sum_{C'\in\mathcal C^{(h)}} \ell^h |C' \cap C \cap \I_h| + O(\frac{d + \ell}{\ell}\log(\ell) |C|)] \\
  = \frac{1}{\ell}\underbrace{\sum_{j\le h\le k} \sum_{C'\in\mathcal C^{(h)}} |C' \cap C|}_{\le \ell |C|} + O\left(\frac{d + \ell}{\ell}\log(\ell) |C|\right)
  = O\left(\frac{d + \ell}{\ell}\log(\ell) |C|\right) .
\end{multline*}
Hence the lemma says that each resource in $C$ is roughly covered $O((d + \ell)/\ell \cdot \log(\ell))$ times by smaller configurations.

We now proceed to prove the lemma by performing the following random experiment
and by Lovasz Local Lemma show that there
is a positive probability of success.
\begin{experiment}
For each $i \in P$ select one configuration $K_i\in\mathcal C_i$ uniformly at random.
\end{experiment}
We write $\mathcal K^{(k)}_i = \{K_i\}$ if $K_i\in\mathcal C^{(k)}_i$ and $\mathcal K^{(k)}_i = \emptyset$ otherwise.
For all $h=0,\dotsc,d$ and $i\in P$
we define the random variable
\begin{equation*}
    X^{(h)}_{i,C} = \sum_{K\in\mathcal K^{(h)}_i} |K \cap C \cap \I_h| \le \min\{3/2 \cdot \ell^4, |C\cap \I_h|\} .
\end{equation*}
Let $X^{(h)}_C = \sum_{i=1}^m X^{(h)}_{i, C}$.
Then
\begin{equation*}
  \E[X^{(h)}_C] \le \frac{1}{\ell} \sum_{C'\in\mathcal C^{(h)}} |C'\cap C\cap \I_h| \le |C\cap \I_h| .
\end{equation*}
We define a set of bad events. As we will show
later, if none of them
occur, the properties from the premise hold.
For each $k$, $C\in\mathcal C^{(k)}$, and $h\le k$ let $B_C^{(h)}$ be the event that
\begin{equation*}
    X_C^{(h)}
    \ge \begin{cases}
      \E[X_C^{(h)}] + 63 |C\cap \I_h| \log(\ell) &\text{ if $k - 5 \le h \le k$}, \\
      \E[X_C^{(h)}] + 135 |C\cap \I_h| \log(\ell) \cdot \ell^{-1} &\text{ if $h \le k - 6$}.
    \end{cases}
\end{equation*}
There is an intuitive reason as to why we define these two different bad events. In the case $h\leq k-6$, we are counting how many times $C$ is intersected by configurations that are much smaller than $C$. Hence the size of this intersection can be written as a sum of independent random variables of value at most $O(\ell^4)$ which is much smaller than the total size of the configuration $|C\cap \I_h|$. Since the random variables are in a much smaller range, Chernoff bounds give much better concentration guarantees and we can afford a very small deviation from the expectation. In the other case, we do not have this property hence we need a bigger deviation to maintain a sufficiently low probability of failure. However, this does not hurt the statement of Lemma~\ref{lma:main-LLL} since we sum this bigger deviation only a constant number of times. With this intuition in mind, we claim the following.
\begin{claim}
For each $k$, $C\in\mathcal C^{(k)}$, and $h\le k$ we have
\begin{equation*}
    \P[B_C^{(h)}] \le \exp\left(- 2\frac{|C \cap \I_h|}{\ell^9} - 18\log(\ell)\right) .
\end{equation*}
\end{claim}

\begin{proof}
Consider first the case that $h \ge k - 5$.
By a Chernoff bound (see Proposition~\ref{cor:chernoff}) with
\begin{equation*}
    \delta = 63\frac{|C\cap \I_h| \log(\ell)}{\E[X_C^{(h)}]} \ge 1
\end{equation*}
we get
\begin{equation*}
    \P[B_C^{(h)}] \le \exp\bigg(-\frac{\delta \E[X^{(h)}_C]}{3 |C\cap \I_h|}\bigg) \le \exp(-21\log(\ell))) \le \exp\bigg(-2 \underbrace{\frac{|C\cap \I_h|}{\ell^{9}}}_{\le 3/2} - 18\log(\ell)\bigg).
\end{equation*}
Now consider $h \le k - 6$.
We apply again a Chernoff bound with
\begin{equation*}
    \delta = 135\frac{|C\cap\I_h| \log(\ell)}{\ell \E[X_C^{(h)}]} \ge \frac{1}{\ell} .
\end{equation*}
This implies
\begin{multline*}
    \mathbb P[B_C^{(h)}]
    \le \exp\left(-\frac{\min\{\delta,\delta^2\} \E[X^{(h)}_C]}{3 \cdot 3/2 \cdot \ell^4}\right)
    \le \exp\left(-30\frac{|C\cap\I_h| \log(\ell)}{\ell^6} \right) \\
    \le \exp\left(-2 \frac{|C\cap \I_h|}{\ell^9} - 18\log(\ell)\right) . \qedhere
\end{multline*}
\end{proof}
\begin{proposition}[Lovasz Local Lemma (LLL)]\label{prop:LLL}
Let $B_1, \dotsc, B_t$ be bad events, and let
$G = (\{B_1,\dotsc,B_t\}, E)$ be a dependency graph for them, in which for every $i$, event $B_i$ is mutually independent of all events $B_j$ for which $(B_i, B_j)\notin E$.
Let $x_i$ for $1\le i \le t$ be such that
$0 < x(B_i) < 1$ and
$\P[B_i]\le x(B_i) \prod_{(B_i,B_j)\in E} (1-x(B_j))$.
Then with positive probability no event $B_i$ holds.
\end{proposition}
Let $k\in\{0,\dotsc,d\}$, $C\in\mathcal C^{(k)}$ and $h\le k$.
For event $B_C^{(h)}$ we set
\begin{equation*}
    x(B_C^{(h)}) = \exp(-|C\cap\I_h| / \ell^9 - 18\log(\ell)) .
\end{equation*}
We now analyze the dependencies of $B_C^{(h)}$.
The event depends only on random variables $K_i$ for a
player $i$ that has a configuration in $\mathcal C^{(h)}_i$ which overlaps with $C\cap \I_h$.
The number of such configurations (in particular,
of such players) is at most
$\ell |C\cap \I_h|$.
Moreover, any player $i'$ of them has only $\ell$ configurations. The random variable $K_{i'}$ only influences
those events $B^{(h')}_{C''}$ where $C' \cap C'' \cap \I_{h'} \neq \emptyset$ for some
some $C'\in\mathcal C_i^{(h')}$.
The number of such events (for player $i'$)
is at most $(3/2)\ell^6$.
Hence, in total $B_C^{(h)}$ depends on $(3/2)|C\cap\I_h|\ell^7<|C\cap\I_h|\ell^8$
other bad events.
We verify the condition for Proposition~\ref{prop:LLL} by calculating
\begin{align*}
    x(B_C^{(h)}) & \prod_{(B_C^{(h)}, B_{C'}^{(h')})\in E} (1 - x(B_{C'}^{(h')})) \\
    &\ge \exp(-|C\cap\I_h|/\ell^9 - 18\log(\ell)) \cdot (1 - \ell^{-18})^{|C\cap\I_h|\ell^8} \\
    &\ge \exp(-|C\cap\I_h|/\ell^9 - 18\log(\ell)) \cdot \exp(- |C\cap\I_h| / \ell^9) \\
    &\ge \exp(-2|C\cap\I_h|/\ell^9 - 18\log(\ell)) \ge \P[B^{(h)}_C] .
\end{align*}
By LLL we have that with positive probability none
of the bad events happen. Let $k\in\{0,\dotsc,d\}$ and $C\in\mathcal C^{(k)}$.
Then for $k - 5 \le h \le k$ we have
\begin{equation*}
  \ell^{h} X^{(h)}_C \le \ell^{h} \E[X_C^{(h)}] + 63 \ell^{h} |C\cap\I_h|\log(\ell)
    \le \ell^{h} \E[X_C^{(h)}] + 95 |C|\log(\ell) .
\end{equation*}
Moreover, for $h\le k-6$ it holds that
\begin{equation*}
  \ell^{h} X^{(h)}_C \le \ell^{h} \E[X_C^{(h)}] + 135 \ell^{h-1} |C\cap\I_h|\log(\ell)
    \le \ell^{h} \E[X_C^{(h)}] + 203 |C|\log(\ell) \cdot \ell^{-1} .
\end{equation*}
We conclude that, for any $0\leq j\leq k$,
\begin{align*}
    \sum_{j\leq h\le k} \sum_{K\in\mathcal K^{(h)}}
 \ell^{h} |K \cap C \cap \I_h|
 &\le \sum_{j\leq h\le k} \ell^{h}  \E[X^{(h)}_{C}] + 1000 \frac{(k-j + 1) + \ell}{\ell} |C| \log(\ell) \\
 &\le \frac{1}{\ell} \sum_{j\leq h\le k} \ell^{h} \sum_{C'\in\mathcal C^{(h)}} |C'\cap C \cap \I_h| + 1000 \frac{d + \ell}{\ell} |C| \log(\ell) .
\end{align*}
This proves Lemma~\ref{lma:main-LLL}.

\begin{remark}{\rm
Since there are at most $\mathrm{poly}(n,m,\ell)$ bad events and each bad event $B$ has $x(B) \ge 1/\mathrm{poly}(n,m,\ell)$, the constructive
variant of LLL by \citet{moser2010constructive} can be applied to find a selection of configuration such that no bad events occur in randomized polynomial time.}
\end{remark}

\subsection{Assignment of resources to configurations}\label{sec:reconstruction}
In this subsection, we show how all the previously established properties allow us to find, in polynomial time, a good assignment of resources to the configurations $\K$ chosen as in the previous subsection. We will denote as in the previous subsection $\K_i^{(k)}=\{K_i\}$ if $K_i\in \C_i^{(k)}$ and $\K_i^{(k)}=\emptyset$ otherwise. We also define $\K^{(k)}=\bigcup_{i}\K_i^{(k)}$ and $\K^{(\geq k)}=\bigcup_{h\geq k}\K^{(k)}$.
Finally we define the parameter
\begin{equation*}
    \gamma = 100.000 \frac{d+\ell}{\ell}\log(\ell) ,
\end{equation*}
which will define how many times each resource can be assigned to configurations in an intermediate solution. Note that $d\le\log(n)/\log(\ell)$. By our choice of $\ell=300.000\log^3(n)$, we have that $\gamma \leq 310.000 \log \log (n)$.
Lemma~\ref{lma:main-LLL} implies the following
bound. For sake of brevity, the proof is deferred to Appendix~\ref{appendix_reconstruct}.
\begin{claim}
\label{cla:reconstruct}
For any $k\geq 0$, any $0\leq j\leq k$, and any $C\in \K^{(k)}$
\begin{equation*}
    \sum_{j\leq h\leq k}\sum_{K\in \K^{(h)}} \ell^{h}|K\cap C \cap \I_h| \leq 2000\frac{d+\ell}{\ell}\log (\ell) |C|
\end{equation*}
\end{claim}

The main technical part of this section is
the following lemma that is proved by induction.
\begin{lemma}
\label{lem:reconstruct}
For any $j\geq 0$, there exists an assignment of resources of $\I_j$ to configurations in $\K^{(\geq j)}$ such that no resource is taken more than $\gamma$ times and each configuration $C\in \K^{(k)}$ ($k\geq j$) receives at least
\begin{equation*}
     \left(1-\frac{1}{\log (n)} \right)^{2(k-j)}\ell^{k-j} |C\cap \I_k|-\frac{3}{\gamma}\sum_{j\leq h\leq k} \sum_{K\in \K^{(h)}} \ell^{h-j}|K\cap C \cap \I_h|
\end{equation*}
resources from $\I_k$.
\end{lemma}

\begin{proof}
We start from the biggest configurations and then iteratively reconstruct a good solution for smaller and smaller configurations. Recall $d$ is the smallest integer such that $\K^{(\geq d)}$ is empty.
Our base case for these configurations in $\K^{(\geq d)}$ is vacuously satisfied. 


Now assume that we have a solution at level $j$, i.e. an assignment of resources to configurations in $\K^{(\geq j)}$ such that no resource is taken more than $\gamma$ times and each configuration $C\in \K^{(k)}$ such that $k\geq j$ receives at least
\begin{equation*}
     \left(1-\frac{1}{\log (n)} \right)^{2(k-j)}\ell^{k-j} |C\cap \I_k|-\frac{3}{\gamma}\sum_{j\leq h\leq k} \sum_{K\in \K^{(h)}} \ell^{h-j}|K\cap C \cap \I_h|
\end{equation*}
resources from $\I_j$.
We show that this implies a solution at level $j-1$ in the following way. First by Lemma~\ref{lma-good-solution}, this implies an assignment of resources of $\I_{j-1}$ to configurations in $\K^{(\geq j)}$ such that each $C\in \K^{(k)}$ receives at least
\begin{align*}
     &\left(1-\frac{1}{\log (n)} \right)\ell \left(\ell^{k-j} \left(1-\frac{1}{\log (n)} \right)^{2(k-j)} |C\cap \I_k|-\frac{3}{\gamma}\sum_{j\leq h\leq k} \sum_{K\in \K^{(h)}} \ell^{h-j}|K\cap C \cap \I_h|\right)\\
     &=\left(1-\frac{1}{\log (n)} \right)^{2(k-(j-1))-1} \ell^{k-(j-1)} |C\cap \I_k|-\frac{3}{\gamma}\left(1-\frac{1}{\log (n)} \right)\sum_{j\leq h\leq k} \sum_{K\in \K^{(h)}} \ell^{h-(j-1)}|K\cap C \cap \I_h|\\
     &\geq  \left(1-\frac{1}{\log (n)} \right)^{2(k-(j-1))-1}\ell^{k-(j-1)} |C\cap \I_k|-\frac{3}{\gamma}\sum_{j\leq h\leq k} \sum_{K\in \K^{(h)}} \ell^{h-(j-1)}|K\cap C \cap \I_h|
\end{align*}
resources and no resource of $\I_{j-1}$ is taken more than $\gamma$ times. Note that we can apply Lemma \ref{lma-good-solution} since we have by Claim \ref{cla:reconstruct} and Lemma \ref{lma-size}
\begin{align*}
        &\left(1-\frac{1}{\log (n)} \right)^{2(k-j)}\ell^{k-j} |C\cap \I_k|-\frac{3}{\gamma}\sum_{j\leq h\leq k} \sum_{K\in \K^{(h)}} \ell^{h-j}|K\cap C \cap \I_h| \\
        &\geq \frac{\ell^{k-j}}{e^2}|C\cap R_k| - \frac{3}{\gamma}2000\ell^{-j}\frac{d+\ell}{\ell}\log(\ell)|C|\\
        &\geq \ell^{-j}|C|\left(\frac{1}{2e^2}-\frac{6000}{\gamma}\frac{d+\ell}{\ell}\log(\ell)\right)\\
        &\geq \frac{\ell^{-j}|C|}{3e^2}>\frac{\ell^3}{1000}
\end{align*}
Now consider configurations in $\K^{(j-1)}$ and proceed for them as follows. Give to each $C\in\K^{(j-1)}$ all the resources in $C\cap \I_{j-1}$ except all the resources that appear in more than $\gamma$ configurations in $\K^{(j-1)}$. Since each deleted resource is counted at least $\gamma$ times in the sum $\sum_{K\in \K^{(j-1)}}|K\cap C\cap \I_{j-1}|$, we have that each configuration $C$ in $\K^{(j-1)}$ receives at least 
\begin{equation*}
    |C\cap \I_{j-1}|-\frac{1}{\gamma}\sum_{K\in \K^{(j-1)}}|K\cap C\cap \I_{j-1}|
\end{equation*}
resources and no resource is taken more than $\gamma$ times by configurations in $\K^{(j-1)}$. Notice that now every resource is taken no more than $\gamma$ times by configurations in $\K^{(\geq j)}$ and no more than $\gamma$ times by configurations in $\K^{(j-1)}$ which in total can sum up to $2\gamma$ times. 

Therefore to finish the proof consider an resource $i\in \I_{j-1}$. This resource is taken $b_i$ times by configurations in $\K^{(\geq j)}$ and $a_i$ times by configurations in $\K^{(j-1)}$. If $a_i+b_i \leq \gamma$, nothing needs to be done. Otherwise, denote by $O$ the set of problematic resources (i.e. resources $i$ such that $a_i+b_i>\gamma$). For every $i\in O$, select uniformly at random $a_i+b_i-\gamma$ configurations in $\K^{(\geq j)}$ that currently contain resource $i$ and delete the resource from these configurations. When this happens, each configuration in $C\in \K^{(\geq j)}$ that contains $i$ has a probability of $(a_i+b_i-\gamma)/b_i$ to be selected to loose this resource. Hence the expected number of resources that $C$ looses with such a process is

\begin{equation*}
    \mu = \sum_{i\in O\cap C} \frac{a_i+b_i-\gamma}{b_i}
\end{equation*}
It is not difficult to prove the following claim. However, for better clarity we defer its proof to appendix \ref{appendix_reconstruct}.
\begin{claim} For any $C\in \K^{(\geq j)}$,
\label{cla:reconstruct_mu}
\begin{equation*}
    \frac{1}{\gamma^2}\sum_{K\in \K^{(j-1)}}|K\cap C \cap \I_{j-1}\cap O|\leq \mu \leq \frac{2}{\gamma} \sum_{K\in \K^{(j-1)}}|K\cap C \cap \I_{j-1}\cap O|
\end{equation*}
\end{claim}
Assume then that $\mu \leq \frac{|C\cap \I_k|}{10^{12}\log^3 (n)}$. Note that $C$ cannot loose more than $\sum_{K\in \K^{(j-1)}}|K\cap C \cap \I_{j-1}\cap O|$ resources in any case. Therefore, by assumption on $\mu$, and since 
\begin{equation*}
    \mu\geq \frac{1}{\gamma^2}\sum_{K\in \K^{(j-1)}}|K\cap C \cap \I_{j-1}\cap O|\ ,
\end{equation*}
we have that 
\begin{equation*}
    \sum_{K\in \K^{(j-1)}}|K\cap C \cap \I_{j-1}\cap O|\leq \frac{\gamma^2}{10^{12}\log^3 (n)} |C\cap \I_k|\leq \frac{10^{11} \log^2 \log (n)}{10^{12}\log^3 (n)}|C\cap \I_k|\leq \frac{1}{\log (n)}|C\cap \I_k|\ .
\end{equation*}
Therefore $C$ looses at most $|C\cap \I_k|/\log (n)$ resources. Otherwise we have that 
\begin{equation*}
    \mu > \frac{|C\cap \I_k|}{10^{12}\log^2 (n)} \geq \frac{\ell^3}{10^{12} \log^3 (n)} \geq 200\log(n)
\end{equation*}
by Lemma~\ref{lma-size}. Hence noting $X$ the number of deleted resources in $C$ we have that 
\begin{equation*}
    \mathbb P\left(X\geq \frac{3}{2}\mu \right) \leq \exp\left(-\frac{\mu}{12} \right)\leq \frac{1}{n^{10}}.
\end{equation*}
With high probability no configuration looses more than 
\begin{equation*}
    \frac{3}{2}\mu \leq \frac{3}{\gamma}\sum_{K\in \K^{(j-1)}}|K\cap C \cap \I_{j-1}\cap O|\leq \frac{3}{\gamma}\sum_{K\in \K^{(j-1)}}|K\cap C \cap \I_{j-1}|
\end{equation*}
resources. Hence each configuration $C\in \K^{(\geq j)}$ ends with at least
\begin{align*}
    &\left(1-\frac{1}{\log (n)} \right)^{2(k-(j-1))-1}\ell^{k-(j-1)} |C\cap \I_k|-\frac{3}{\gamma}\sum_{j\leq h\leq k} \sum_{K\in \K^{(h)}} \ell^{h-(j-1)}|K\cap C \cap \I_h|\\
    &-\frac{1}{\log (n)}\left(1-\frac{1}{\log (n)} \right)^{2(k-(j-1))-1}\ell^{k-(j-1)} |C\cap \I_k| - \frac{3}{\gamma}\sum_{K\in \K^{(j-1)}}|K\cap C \cap \I_{j-1}|\\
    &\geq \left(1-\frac{1}{\log (n)} \right)^{2(k-(j-1))}\ell^{k-(j-1)} |C\cap \I_k|-\frac{3}{\gamma}\sum_{j-1\leq h\leq k} \sum_{K\in \K^{(h)}} \ell^{h-(j-1)}|K\cap C \cap \I_h|
\end{align*}
resources which concludes the proof.
\end{proof}
\begin{corollary}
\label{reconstruct_corollary}
There exists an assignment of resources $\I$ to $\K$ such that each configuration $C\in  \K$ receives at least $\left\lfloor |C|/(100\gamma) \right\rfloor$ resources. Moreover, this assignment can be found in polynomial time.
\end{corollary}
\begin{proof}
Lemma \ref{lem:reconstruct} for $k=0$ and Claim \ref{cla:reconstruct} together imply that
we can assign at least
\begin{equation*}
    \frac{|C|}{2e^2}-\frac{6000}{100.000}|C|\geq \frac{|C|}{100}
\end{equation*}
resources to every $C\in \K$ such that no resource in $\I$ is assigned more than $\gamma$ times. In particular, we can fractionally assign
at least $|C| / (100\gamma)$ resources
to each $C\in \K$ such that no resource
is assigned more than once.
By integrality of the bipartite matching polytope, the corollary follows.
\end{proof}

\section{Combinatorial Santa Claus with large hyperedges}
We now show how the previous result implies a $O(\log\log(n))$-approximation algorithm for the Combinatorial Santa Claus problem
when all hyperedges are sufficiently large.
More precisely, we show that we can either decide that there is no $\delta$-relaxed perfect matching or find
a $(\delta / \alpha)$-relaxed perfect matching for $\alpha = O(\log\log(n))$, when all hyperedges are of size at
least $\alpha / \delta$.

First consider $\delta = 1$.
We write the natural linear program for the bipartite hypergraph matching problem. We have a variable $x_C$ for every edge $C\in\C$ which tells us if it is taken or not. We require every vertex $i\in P$ to have exactly one configuration incident to it and every vertex $j \in \I$ to be contained in at most one configuration.
\begin{align*}
     \sum_{C\in \C : i \in C} & x_C = 1 \quad \text{ for all } i \in P \\
     \sum_{C\in \C : j \in C} & x_C \leq 1 \quad \text{ for all } j \in \I \\
    &  x_C \geq 0 \quad \text{ for all } C \in \C
\end{align*}
Now alter $(P\cup R, \C)$ by replacing each edge $C\in \C$ with $x_C \cdot T$ copies where $T$ is the lowest common multiple all denominators of $x_{C'}$, $C'\in\C$. In the resulting hypergraph each player appears in exactly $T$ hyperedges and every resource in at most $T$ hyperedges.
Now, a $(1/\alpha)$-relaxed matching can be found using the result in the previous section about regular hypergraphs.
Indeed, the intermediate hypergraph could have exponential size, but this can be avoided by performing the preprocessing in Appendix~\ref{appendix_main} without constructing the graph explicitly.

To extend this to general $\delta$ one can consider same LP for a new bipartite hypergraph which has the same set of vertices, but every edge $C \in \C$ incident on any $i \in P$ is replaced by all subsets of $C$ that have size at least $\delta|C|$.
Although this LP has exponential size, it can be solved efficiently using a separation oracle for its dual
(we refer to~\cite{BansalSrividenko} for details).

\section{On the connection to the Santa Claus problem}
In this section, we show how the Santa Claus problem and the Combinatorial Santa Claus problem are almost
equivalent in terms of approximability.
We show that the Combinatorial Santa Claus problem is a special case of the Santa Claus problem, hence preserving the
approximation rate of any algorithm. On the other hand,
given a $c$-approximation for the Combinatorial Santa Claus problem, we get a $O((c\log^*(n))^2)$-approximation
for the Santa Claus problem. We note that there is a reduction similar to the latter in~\cite{DBLP:conf/focs/ChakrabartyCK09},
but it loses a factor of $O(\log(n))$.

\subsection{Combinatorial Santa Claus to Santa Claus}
The idea in this reduction is to replace each player by a set of players, one for each of the $t$ configuration containing him.
These players will share together $t-1$ large new resources, but to satisfy all, one of them has to get other resources, which
are the original resources in the corresponding configuration.
\begin{description}
\item[Players.]
For every vertex $v \in P$, and every hyperedge $C \in \C$ that $v$ belongs to, we create a player $p_{v,C}$ in the Santa Claus instance. 
\item[Resources.]
For every vertex $u \in \I$, create a resource $r_{u}$ in the Santa Claus instance. 
For any vertex $v \in P$ such that it belongs to $t$ edges in $\C$, create $t-1$ resources $r_{v,1}, r_{v,2}, \ldots, r_{v,t-1} $. 
\item[Values.]
For any resource $r_{u}$ for some $u \in \I$ and any player $p_{v,C}$ for some $C \in \C$, the resource has a value $\frac{1}{|C|-1}$ if $u \in C$, otherwise it has value $0$. Any resource $r_{v,i}$ for some $v \in P$ and $i \in \mathbb N$, has value $1$ for any player $p_{v,C}$ for some $C \in \C$ and $0$ to all other players. 
\end{description}
It is easy to see that given an $\alpha$-relaxed matching in the original instance, one can construct an $\alpha$-approximate solution for the Santa Claus instance.

For the other direction, notice that for each $v \in P$, there exists a player $p_{v,C}$ for some $C \in \C$, such that it gets resources only of the type $r_{u}$. One can simply assign the resource $u \in \I$ to the player $v$ for any resource $r_{u}$ assigned to $p_{v,C}$.

\subsection{Santa Claus to Combinatorial Santa Claus}
This subsection is devoted to the proof of the following theorem.
\begin{theorem}
A $c$-approximation algorithm to the Combinatorial Santa Claus problem yields an
$O((c\log^* (n))^2)$-approximation algorithm to the Santa Claus problem.
\end{theorem}
\begin{proof}
We write $(\log)^k(n) = \underbrace{\log \cdots \log}_{\times k}(n)$ and $(\log)^0(n) = n$.   

\paragraph*{Construction.}
We describe how to construct a hypergraph matching instance from a Santa Claus instance in four steps by reducing to the following more and more special cases.

\begin{description}
    \item[Geometric grouping.] In this step, given arbitrary $v_{ij}$, we reduce it to an instance such that $\OPT = 1$ and for each $i, j$ we have $v_{ij} = 2^{-k}$ for some integer $k$ and $1/(2n) < v_{ij} \le 1$.
    This step follows easily from guessing $\OPT$, rounding down the sizes, and omitting all small elements in a solution.
    \item[Reduction to O(log*(n)) size ranges.] Next, we reduce to an instance such that for each player $i$ there is some $k \le \log^*(2n)$ such that for each resource $j$, $v_{ij}\in\{0, 1\}$ or $1/(\log)^k(2n) < v_{ij} \le 1/(\log)^{k+1}(2n)$. We explain this step below.
    
    Each player and resource is copied to the new instance.
However, we will also add auxiliary players and resources.
Let $i$ be a player.
In the optimal solution there is some
$0 \le k \le \log^*(2n)$ such that the values of all resources
$j$ with $1/(\log)^k(2n) < v_{ij} \le 1/(\log)^{k+1}(2n)$ assigned
to player $i$ sum up to at least $1/\log^*(2n)$.
Hence, we create $\log^*(2n)$ auxiliary players which correspond to
each $k$ and each of which share an resource with the original player that
has value $1$ for both.
The original player needs to get one of these resources, which means
one of the auxiliary players needs to get a significant value from
the resources with $1/(\log)^k(2n) < v_{ij} \le 1/(\log)^{k+1}(2n)$.
This reduction loses a factor of at most $\log^*(2n)$.
Hence, $\OPT \geq 1/\log^*(2n)$.

\item[Reduction to 3 sizes.] We further reduce to an instance such that for each player $i$ there is some value $v_i$ such that for each resource $j$, $v_{ij}\in\{0, v_i, 1\}$. 

Let $i$ be some player who has only resources of value $v_{ij}\in\{0,1\}$ or
$1/(\log)^k(2n) < v_{ij} \le 1/(\log)^{k+1}(2n)$. 
There are at most $\log((\log)^k(2n)) \leq (\log)^{k+1}(2n)$ distinct values of
the latter kind. The idea is to assign bundles of resources of value $1/\left( \log^*(2n)(\log)^{k+1}(2n) \right)$ to the player $i$. For each distinct value, we create sufficiently many (say, $2n$)
auxiliary players. These auxiliary players
each share a new resource with $i$, which has value $1$ for this
player and value $1/\left( \log^*(2n)(\log)^{k+1}(2n) \right)$ for $i$.
If $i$ takes such an resource, the auxiliary player
should collect a value of $0.5/\left( \log^*(2n)(\log)^{k+1}(2n) \right)$ of resources of his
particular value. Hence, we set the values for these resources
for this player to $v_{ij} / \left(0.5/\left(\log^*(2n)(\log)^{k+1}(2n) \right) \right)$. We lose a factor of $O(\log^*(2n))$ in this step as there are at least $\lceil (0.5(\log)^{k+1}(2n))/\log^*(2n) \rceil$ bundles of size at least $0.5/\left( \log^*(2n)(\log)^{k+1}(2n) \right)$ for player $i$. Now rescale the instance appropriately to get $\OPT=1$. 

\item[Reduction to Combinatorial Santa Claus.]
For each player create a vertex in $P$ and for
each resource create a vertex in $\I$.
Moreover, for every player $i$, add $1/v_i$ vertices to $P$ and
the same number to $\I$. 
Add one hyperedge for each resource he values at $1$ (containing $i$ and
this resource).
Add another hyperedge for $i$ containing $i$ and all $1/v_i$ new vertices
in $\I$.
Pair the new vertices (one from $\I$ and one from $P$)
and add those pairs as hyperedge.
Finally, for each new vertex in $P$ and
each resource that $i$ values at $v_i$, add a hyperedge containing them.
This reduction does not lose any factor.

\end{description}

\paragraph*{Correctness.} 
Steps (1) and (2) are easy and we omit it. For the (3) step, notice that player $i$ might only be able to get a value of $O(1/(\log^*(2n))^2)$. Hence, we lose at most a factor of $O(\log^*(2n))$. It is not hard to see that one can reconstruct a solution to the instance produced by step ($2$) given a solution to instance produced by step ($3$) and vice-versa. 

For the (4) step, notice that if the player $i$ of value $v_i$ in the Santa Claus instance takes a resource of value 1, then the new vertices of $P$ and $R$ can form a matching and the vertex $i$ can take the same resource. On the other hand, if there exists a 1-relaxed matching in the hypergraph matching instance, we argue as follows. If the vertex $i$ takes an edge with single vertex, the player $i$ is given the corresponding resource of value 1. Otherwise, vertex $i$ must have taken the edge with $1/v_i$ new vertices corresponding to $i$ and consequently $1/v_i$ new vertices in $P$ must have taken edges with vertices corresponding to resources of value $v_i$ each. We simply assign the resources corresponding to these vertices to the player $i$. 

Hence, any $c$-approximate solution to the Combinatorial Santa Claus problem yields a $O((c\log^*(n))^2)$ solution to the Santa Claus problem. We lose a factor of $O(1)$ in step (1), $O(\log^*(n))$ factors each in steps (2) and (3) due to the reduction and factors of $O(c)$ each in steps (3) and (4) while reconstructing the solution.    
\end{proof}

\section{Conclusion}

We formulated a new matching problem in non-uniform hypergraphs that generalizes the case of uniform hypergraphs that has been already studied in the context of the restricted Santa Claus problem. Under the assumption that the hypergraph is regular and all edges are sufficiently large, we proved that there is always a $(1/\alpha)$-relaxed perfect matching for $\alpha = O(\log \log (n))$. This result generalizes the work of \citet{BansalSrividenko}.
It remains an intriguing question whether one can get $\alpha = O(1)$ as it is possible in the uniform case.
One idea (similar to Feige's proof in the uniform case~\cite{Feige}) is to to view our proof as a sparsification theorem
and to apply it several times.
Given a set of hyperedges such that every player has $\ell$ hyperedges and every resource appears in no more than $\ell$ hyperedges, one would like to select $\textrm{polylog}(\ell)$ hyperedges for each player such that all resources appear in no more than $\textrm{polylog}(\ell)$ of the selected hyperedges. It is not difficult to see than our proof actually achieves this when $\ell=\textrm{polylog}(n)$. However, repeating this after the first step seems to require new ideas since our bound on the number of times each resource is taken is $\Omega \left(\frac{d+\ell}{\ell}\log(\ell) \right)$ where $\ell$ is the current sparsity and $d$ the number of configuration sizes. For the first step, we conveniently have that $d=O(\log (n))=O(\ell)$ but after the first sparsification, it may not be true.

We also provided a reduction from Santa Claus to the Combinatorial Santa Claus.
An interesting result is to improve the $O(\log^*(n))^2$ factor in the reduction to a constant.

Finally, we gave a positive result for the Combinatorial Santa Claus problem when all hyperedges are large.
It is not clear whether this assumption is necessary. In other words, knowing that a perfect matching exists,
can one compute an $(1/\alpha)$-relaxed perfect matching (for some small $\alpha$)
even when edges are allowed to be very small (e.g. of size $2$)?
This is a particularly interesting question as it would have immediate consequences for the Santa Claus problem.
However, such an algorithm cannot be based on regularity, as this is not a sufficient condition for the existence of a
$(1/\alpha)$-relaxed matching when small hyperedges are allowed, as demonstrated by the example in the beginning of the paper.

\section{Acknowledgements}
The authors wish to thank Ola Svensson for helpful discussions on the problem.

\bibliographystyle{plainnat}
\bibliography{refs}

\appendix

\section{Concentration bounds}
\begin{proposition}[Chernoff bounds (see e.g.~\cite{mitzenmacher2017probability})]
\label{chernoff}
Let $X=\sum_i X_i$ be a sum of independent random variables such that each $X_i$ can take values in a range $[0,1]$. Define $\mu=\mathbb E(X)$. We then have the following bounds 

\begin{equation*}
    \mathbb P \left(X\geq (1+\delta)\mathbb E(X) \right) \leq \exp\left(-\frac{\min\{\delta,\delta^2\} \mu}{3} \right)
\end{equation*} for any $\delta>0$.
\begin{equation*}
    \mathbb P \left(X\leq (1-\delta)\mathbb E(X) \right) \leq \exp\left(-\frac{\delta^2 \mu}{2} \right)
\end{equation*} for any $0<\delta<1$.
\end{proposition}

The following proposition follows immediately from Proposition \ref{chernoff} 
by apply it with $X'=X/a$.
\begin{proposition}
\label{cor:chernoff}
Let $X=\sum_i X_i$ be a sum of independent random variables such that each $X_i$ can take values in a range $[0,a]$ for some $a>0$. Define $\mu=\mathbb E(X)$. We then have the following bounds 

\begin{equation*}
    \mathbb P \left(X\geq (1+\delta)\mathbb E(X) \right) \leq \exp\left(-\frac{\min\{\delta,\delta^2\} \mu}{3a} \right)
\end{equation*} for any $\delta>0$.
\begin{equation*}
    \mathbb P \left(X\leq (1-\delta)\mathbb E(X) \right) \leq \exp\left(-\frac{\delta^2 \mu}{2a} \right)
\end{equation*} for any $0<\delta<1$.
\end{proposition}

\section{Assuming $\ell = \textrm{poly(log} (n))$}\label{appendix_main}
When considering $\alpha$-relaxed perfect matchings
in regular hypergraphs with sufficiently large hyperedges, we can assume that $\ell =300.000\log^{3}(n)$ at a constant loss:

If $\ell$ is smaller than $300.000\log^{3}(n)$, then we simply duplicate all hyperedges an appropriate number of times.
If $\ell$ is larger, we select for each player $300.000\log^{3}(n)$ configurations uniformly at random from his configurations. The expected number of times a resource appears in a configuration with this process is at most $300.000\log^{3}(n)$. Hence, the probability that a resource appears more than $600.000\log^{3}(n)$ times is at most $\exp \left(- 1/3 \cdot 100.000\log^{3}(n)\right)\leq 1/n^{10}$ by a standard Chernoff bound (see Proposition~\ref{cor:chernoff}). Hence with high probability this event does not happen for any resource.
We now have that each player has $300.000\log^{3}(n)$ configurations and each resource does not appear in more than $600.000\log^{3}(n)$ configurations.
Taking for each configuration $C$ only $\lfloor |C| / 2 \rfloor$ resources we can reduce the
latter bound to $300.000\log^3(n)$ as well:
The previous argument gives a half-integral matching of resources to configurations satisfying
the mentioned guarantee. Then by integrality
of the bipartite matching polytope there
is also an integral one.

\section{Omitted proofs from Section~\ref{sec:sequence}}\label{appendix_sequence}
\begin{customthm}{\ref{lma-size}}(restated)
Consider Random Experiment~\ref{exp:sequence}
with $\ell\geq 300.000\log^{3} (n)$.
For any $k\geq 0$ and any $C\in\C^{(\geq k)}$ we have 
  \begin{equation*}
      \frac{1}{2} \ell^{-k}|C| \le |\I_k \cap C| \le \frac{3}{2} \ell^{-k}|C|
  \end{equation*}
with probability at least $1-1/n^{10}$.
\end{customthm}
\begin{proof}
 The lemma trivially holds for $k=0$. 
 For $k>0$, by assumption $C\in\C^{(\geq k)}$ hence $|C|\geq \ell^{k+3}$. Since each resource of $\I=\I_0$ survives in $\I_k$ with probability $\ell^{-k}$ we clearly have that in expectation
 \begin{equation*}
     \mathbb E(|\I_k\cap C|) =  \ell^{-k}|C|
 \end{equation*}
 Hence the random variable $X=|\I_k\cap C|$ is a sum of independent variables of value either $0$ or $1$ and such that $\mathbb E (X)\geq \ell^3$. By a standard Chernoff bound (see Proposition \ref{cor:chernoff}), we get
 \begin{equation*}
     \mathbb P\left(X\notin \left[\frac{\mathbb{E}(X)}{2}, \frac{3\mathbb{E}(X)}{2}\right]\right) \leq 2 \exp \left(-\frac{\mathbb E(X)}{12} \right) \leq 2 \exp \left(-\frac{300.000\log^3 (n)}{12} \right) \leq  \frac{1}{n^{10}}
 \end{equation*}
 since by assumption $\ell \geq 300.000\log^3 (n)$.
\end{proof}

\begin{customthm}{\ref{lma-overlap-representative}}(restated)
Consider Random Experiment~\ref{exp:sequence}
with $\ell\geq 300.000\log^{3} (n)$.
For any $k\geq 0$ and any $C\in\C^{(\geq k)}$ we have 
  \begin{equation*}
      \sum_{C'\in \C^{(k)}} |C'\cap C\cap \I_k| \leq \frac{10}{\ell^{k}} \left(|C|+\sum_{C'\in \C^{(k)}} |C'\cap C| \right)
  \end{equation*}
with probability at least $1-1/n^{10}$.
\end{customthm}
\begin{proof} The expected value of the random variable $X=\sum_{C'\in \C^{(k)}} |C'\cap C\cap \I_k|$ is 
\begin{equation*}
    \mathbb E(X) = \frac{1}{\ell^k} \sum_{C'\in \C^{(k)}} |C'\cap C|.
\end{equation*}
Since each resource is in at most $\ell$ configurations, $X$ is a sum of independent random variables that take value in a range $[0,\ell]$. Then by a standard Chernoff bound (see Proposition \ref{cor:chernoff}), we get

\begin{equation*}
    \mathbb P\left(X\ge  10  \left(\frac{|C|}{\ell^k} + \mathbb E(X)\right) \right) \leq \exp\left(-\frac{3|C|}{\ell^{k+1}}\right) \leq \frac{1}{n^{10}} ,
\end{equation*}
since by assumption, $|C|\geq \ell^{k+3}$ and $\ell \geq 300.000\log ^3(n)$.

\end{proof}

We finish by the proof of the last property. As mentioned in the main body of the paper, this statement is a generalization of some ideas that already appeared in \cite{BansalSrividenko}. However, in \cite{BansalSrividenko}, the situation is simpler since they need to sample down the resource set only once (i.e. there are only two sets $R_1\subseteq R$ and not a full hierarchy of resource sets $R_d\subseteq R_{d-1}\subseteq \cdots \subseteq R_1 \subseteq R$). Given the resource set $R_1$, they want to select configurations and give to each selected configuration $K$ all of its resource set $|K\cap R_1|$ so that no resource is assigned too many times. In our case the situation is also slightly more complex than that since at every step the selected configurations receive only a fraction of their current resource set. Nevertheless, we extend the ideas of \citet{BansalSrividenko} to our more general setting. We recall the main statement before proceeding to its proof.
\begin{customthm}{\ref{lma-good-solution}}(restated)
Consider Random Experiment~\ref{exp:sequence}
with $\ell\geq 300.000\log^{3} (n)$.
Assume that the bounds in Lemma~\ref{lma-size} hold
for some $k\ge 0$.
Then with probability at least $1 - 1/n^{10}$
the following holds for all
$\F\subseteq \C^{(\geq k+1)}$, $\alpha:\F \rightarrow \mathbb N$, and $\gamma \in\mathbb N$ such that $\ell^3/1000\leq \alpha(C) \leq n $ for all $C\in\F$ and
$\gamma \in \{1,\dotsc,\ell\}$:
If there is a $(\alpha,\gamma)$-good assignment of $\I_{k+1}$ to $\F$, then there is a $(\alpha',\gamma)$-good assignment of $\I_k$ to $\F$ where
\begin{equation}\label{property:3}
    \alpha'(C) \ge \ell \left(1-\frac{1}{\log (n)} \right) \alpha(C)
\end{equation}
for all $C\in\F$.
Moreover, this assignment can be found in polynomial time.
\end{customthm}

We first provide the definitions of a flow network that allows us to state a clean condition whether a good assignment of resources exists or not. We then provide the high probability statements that imply the lemma.

For any subset of configurations $\mathcal F \subseteq \C^{(\geq k+1)}$, resource set $\I_k$, $\alpha:\F \rightarrow \mathbb N$, and any integer $\gamma$, consider the following directed network (denoted by $\mathcal N (\mathcal F, \I_k, \alpha,\gamma)$). Create a vertex for each configuration in $\mathcal F$ as well as a vertex for each resource. Add a source $s$ and sink $t$. Then add a directed arc from $s$ to the vertex 
$C\in\mathcal F$ with capacity $\alpha(C)$. For every pair of a configuration $C$ and a resource $i$ such that $i\in C$ add a directed arc from $C$ to $i$ with capacity $1$. Finally, add a directed arc from every resource to the sink of capacity $\gamma$. See Figure \ref{fig:network_flow} for an illustration. 

\begin{figure}[]
    \centering
    \includegraphics[scale=1.15]{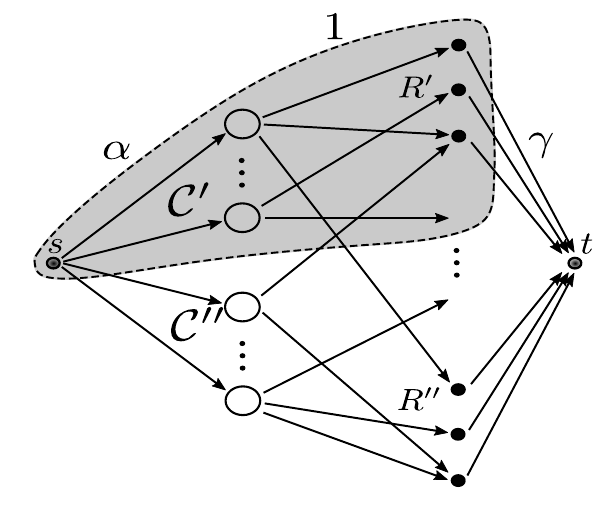}
    \caption{The directed network and an $s$-$t$ cut}
    \label{fig:network_flow}
\end{figure}

We denote by 
\begin{equation*}
    \textrm{maxflow}\left(\mathcal N (\mathcal F, \I_k, \alpha,\gamma)\right)
\end{equation*}
the value of the maximum $s$-$t$ flow in $\mathcal N (\mathcal F, \I_k, \alpha,\gamma)$.

\begin{lemma}\label{lem:flow_black_box}
Let $\mathcal F$ be a set of configurations,
$\I' \subseteq \I$, $\alpha:\F \rightarrow \mathbb N$ a set of resources, $\gamma\in\mathbb N$, and $\epsilon \ge 0$.
Define
\begin{equation*}
    \alpha'(C) = \lfloor (1 - \epsilon) \alpha(C) \rfloor .
\end{equation*}
There is an $(\alpha', \gamma)$-good assignment of $\I'$ to $\F$ if and only if for every $\F' \subseteq \F$, the maximum flow in the network $\mathcal N (\mathcal F', \I', \alpha,\gamma)$ is of value at least $\sum_{C\in \F'}\alpha'(C)$. Moreover, this assignment can be found in polynomial time.
\end{lemma}

\begin{proof}
First assume there is such an $(\alpha', \gamma)$-good assignment. Then send a flow of $\alpha'(C)$ from $s$ to each $C\in \F$. If resource $i$ is assigned to $C$, send a flow of $1$ from $C$ to $i$. Finally ensure that flow is preserved at every vertex corresponding to an resource by sending the correct amount of flow to $t$. Since no resource is taken more than $\gamma$ times, this flow is feasible. 

We prove the other direction by contradiction. Denote by $\N$ the network $\N(\F,\I',\alpha',\gamma)$.
If there is no good assignment satisfying the condition of the lemma then the maximum flow in $\N$ must be strictly less than $\sum_{C\in \F}\alpha'(C)$ (otherwise consider the maximum flow, which can be taken to be integral, and give to every configuration $C$ all the resources to which they send a flow of $1$). Then by the max-flow min-cut theorem, there exists an $s$-$t$ cut $S$ that has value strictly less than $\sum_{C\in \F}\alpha'(C)$. Let $\mathcal C'$ be the set of configurations on the side of the source in $S$. Notice that $\mathcal C'$ cannot be empty by assumption on the value of the cut.

Consider the induced network $\mathcal N (\mathcal C', \I', \alpha',\gamma)$ and the cut $S$ in it. It has a value strictly lower than $\sum_{C\in \mathcal C'} \alpha'(C)$. This, in turn implies that the cut $S$ in $\mathcal N (\mathcal C', \I', \alpha,\gamma)$ has a value strictly lower than $\sum_{C\in \mathcal C'} \alpha'(C)$, since this cut does not contain any edge from the source $s$ to some configuration. Hence the maximum flow in $\mathcal N (\mathcal C', \I', \alpha,\gamma)$ has a value strictly less than $\sum_{C\in \mathcal C'} \alpha'(C)$,
a contradiction to the assumption in the premise. 
\end{proof}

\begin{lemma}
\label{lem:flow_conservation}
Let $\mathcal F\subseteq \mathcal C^{\geq (k+1)}$, $\alpha:\F \rightarrow \mathbb N$ such that $\ell^3/1000 \leq \alpha(C) \leq n$ for all $C\in\F$, and $1 \le \gamma \le \ell$. Denote by $\mathcal N$ the network $\mathcal N (\mathcal F, \I_k, \ell \cdot \alpha,\gamma)$ and by $\Tilde{\mathcal{N}}$ the network $\mathcal N (\mathcal F, \I_{k+1},\alpha,\gamma)$. Then
 \begin{equation*}
     \mathrm{maxflow}\left(\mathcal{N}\right)\geq \frac{\ell}{1+0.5/\log (n)}  \mathrm{maxflow}\left(\Tilde{\mathcal{N}}\right)
 \end{equation*}
 with probability at least $1-1/(n\ell)^{20|\mathcal F|}$.
\end{lemma}
\begin{proof}
We use the max-flow min-cut theorem that asserts that the value of the maximum flow in a network is equal to the value of the minimum $s$-$t$ cut in the network. Consider a minimum cut $S$ of network $\mathcal N$ with $s\in S$ and $t\notin S$. Denote by $c(S)$ the value of the cut. We will argue that with high probability this cut induces a cut of value at most $c(S) / \ell \cdot (1+0.5/\log(n))$
in the network $\Tilde{\mathcal N}$.
This directly implies the lemma.

Denote by $\mathcal C'$ the set of configurations of $\F$ that are in $S$, i.e., on the source side of the cut, and $\mathcal C''=\mathcal{F}\setminus \mathcal C'$.
Similarly consider $\I'$ the set of resources in the $s$ side of the cut and $\I''= \I_k\setminus \I'$. With a similar notation, we denote $\Tilde \I' = \I'\cap \I_{k+1}$ the set of resources of $\I'$ surviving in $\I_{k+1}$; and $\Tilde \I'' = \I''\cap \I_{k+1}$. Finally, denote by $\Tilde S$ the cut in $\Tilde{\mathcal N}$ obtained by removing resources of $R'$ that do not survive in $\I_{k+1}$ from $S$, i.e.,
$\Tilde S = \{s\}\cup \mathcal C' \cup \I'$.
The value of the cut $S$ of $\mathcal N$ is
\begin{equation*}
    c(S) = \sum_{C\in \mathcal C''} \ell \cdot \alpha(C) + e(\mathcal C',\I'')+ \gamma  |\I'|
\end{equation*}
where $e(X,Y)$ denotes the number of edges from $X$ to $Y$.
The value of the cut $\Tilde S$ in $\Tilde{\mathcal N}$ is
\begin{equation*}
    c( \Tilde S) = \sum_{C\in \mathcal C''} \alpha(C) + e(\mathcal C',\Tilde \I'')+ \gamma  |\Tilde \I'|
\end{equation*}
We claim the following properties.
\begin{claim}
\label{cla:size_configurations}
For every $C\in \mathcal F$, the outdegree of the vertex corresponding to $C$ in $\mathcal N$ is at least $\ell^4/2$.
\end{claim}
Since $C\in \C^{(\geq k+1)}$ and by Lemma \ref{lma-size}, we clearly have that $|C\cap \I_k|\geq \ell^4/2$.
\begin{claim}
\label{cla:size_cut}
It holds that
\begin{equation*}
 c(S)\geq \frac{|\F| \ell^3}{1000} .
\end{equation*}
\end{claim}
We have by assumption on $\alpha(C)$
\begin{multline*}
    c(S) = \sum_{C\in \mathcal C''} \ell \cdot \alpha(C) + e(\mathcal C',\I'')+ \gamma |\I'|
    \geq \sum_{C\in \mathcal C''} \frac{\ell^3}{1000} + e(\mathcal C',\I'')+ \gamma |\I'|\\
    \geq \frac{|\mathcal C''|\ell^3}{1000} + e(\mathcal C',\I'')+ \gamma |\I'|
\end{multline*}
Now consider the case where $e(\mathcal C',\I'')\leq |\mathcal C'|\ell^3 / 1000$.
Since each vertex in $\mathcal C'$ has outdegree at least $\ell^4/2$ in the network $\mathcal N$ (by Claim~\ref{cla:size_configurations}) it must be that $e(\mathcal C',\I')\geq |\mathcal C'|\ell^4 / 2 - |\mathcal C'|\ell^3 / 1000 > |\mathcal C'|\ell^4 / 3$.
Using that each vertex in $\I'$ has indegree at most $\ell$ (each resource is in at most $\ell$ configurations), this implies
$|\I'|\geq |\mathcal C'|\ell^3 / 3$. Since $\gamma \geq 1$ we have in all cases that $e(\mathcal C',\I'')+ \gamma  |\I'|\geq |\mathcal C'|\ell^3 / 1000$. Hence 
\begin{equation*}
    c(S) \geq \frac{|\mathcal C''|\ell^3}{1000} + \frac{|\mathcal C'|\ell^3}{1000} =  \frac{|\F| \ell^3}{1000} .
\end{equation*}
This proves Claim~\ref{cla:size_cut}.
We can now finish the proof of the lemma.
Denote by $X$ the value of the random variable $e(\mathcal C',\Tilde{\I''})+ \gamma  |\Tilde{\I'}|$. We have that 
\begin{equation*}
    \mathbb E[X] = \frac{1}{\ell}(e(\mathcal C',\I'')+ \gamma  |\I'|).
\end{equation*}
Moreover, $X$ can be written as a sum of independent variables in the range $[0, \ell]$ since each vertex is in at most $\ell$ configurations and $\gamma \le \ell$ by assumption. By a Chernoff bound (see Proposition \ref{cor:chernoff}) with
\begin{equation*}
    \delta = \frac{0.5 c(S)}{\log(n) \cdot (c(S)-\sum_{C\in \mathcal C''} \alpha(C))} \geq \frac{0.5}{\log(n)}
\end{equation*} 
we have that
\begin{multline*}
    \mathbb P\left(X\geq \mathbb E(X)+\frac{0.5 c(S)}{\ell\log(n)}\right) 
    \leq \exp\left(-\frac{\min\{\delta,\delta^2\}\mathbb E(X)}{3\ell} \right) \\
    \leq \exp\left(-\frac{c(S)}{12\ell^2\log^2 (n)} \right)
    \leq \exp\left(-\frac{|\mathcal F|\ell^3}{12.000\ell^2\log^2 (n)} \right)
    \leq \frac{1}{(n\ell)^{20|\F|}} ,
\end{multline*}
where the third inequality comes from Claim~\ref{cla:size_cut} and the last one from the assumption that $\ell\geq 300.000\log^{3}(n)$.
Hence with probability at least $1-1/(n\ell)^{20|\F|}$, we have that 
\begin{equation*}
    c( \Tilde S) = \sum_{C\in \mathcal C''} \alpha(C) + e(\mathcal C',\Tilde \I'')+ \gamma |\Tilde \I'| \leq \frac{1}{\ell}c(S)+\frac{0.5}{\ell \log (n)}c(S) .\qedhere
\end{equation*}
\end{proof}
We are now ready to prove Lemma~\ref{lma-good-solution}.
Note that Lemma \ref{lem:flow_conservation} holds with probability at least $1-1/(n\ell)^{20|\F|}$.
Given the resource set $\I_k$ and a cardinality $s = |\F|$ there are $O((n\ell)^{2s})$ ways of defining a network satisfying the conditions from Lemma~\ref{lem:flow_conservation} ($(m\ell)^s\le (n\ell)^s$ choices of $\F$, $n^{s}$ choices for $\alpha$ and $\ell$ choices for $\gamma$). By a union bound, we can assume that the properties of Lemma~\ref{lem:flow_conservation} hold for every possible network with probability at least $1 - 1/n^{10}$.
Assume now there is a $(\alpha,\gamma)$-good assignment of $\I_{k+1}$ to some family $\F$. Then by Lemma~\ref{lem:flow_black_box} the $\mathrm{maxflow}(\N(\F',\I_{k+1}, \alpha,\gamma))$ is exactly $\sum_{C\in \F'}\alpha(C)$ for any $\F'\subseteq \F$. By Lemma~\ref{lem:flow_conservation}, this implies that $\mathrm{maxflow}(\N(\F',\I_{k}, \ell \cdot \alpha,\gamma))$ is at least $\ell/(1+0.5/\log(n)) \sum_{C\in \F'}\alpha(C)$. By Lemma \ref{lem:flow_black_box}, this implies a $(\alpha',\gamma)$-good assignment from $\I_k$ to $\F$, where
\begin{equation*}
    \alpha'(C) = \lfloor\ell/(1+0.5/\log(n))\rfloor \alpha(C) \ge \ell / (1 + 1/\log(n)) \alpha(C) \geq \ell(1 - 1/\log(n)) \alpha(C).
\end{equation*}

\section{Omitted proofs from Section~\ref{sec:reconstruction}}\label{appendix_reconstruct}
\begin{customcla}{\ref{cla:reconstruct}}(restated)
For any $k\geq 0$, any $0\leq j\leq k$, and any $C\in \K^{(k)}$
\begin{equation*}
    \sum_{j\leq h\leq k}\sum_{K\in \K^{(h)}} \ell^{h}|K\cap C \cap \I_h| \leq 2000\frac{d+\ell}{\ell}\log (\ell) |C|.
\end{equation*}
\end{customcla}

\begin{proof}[Proof of Claim \ref{cla:reconstruct}]
By Lemma~\ref{lma:main-LLL} we have that 
\begin{equation*}
    \sum_{j\leq h\leq k}\sum_{K\in \K^{(h)}} \ell^{h}|K\cap C \cap \I_h| \leq \frac{1}{\ell} \sum_{j\leq h\leq k}\sum_{C'\in \C^{(h)}} \ell^{h}|C'\cap C \cap \I_h| + 1000\frac{d+\ell}{\ell}\log (\ell) |C|.
\end{equation*}
Furthermore, by Lemma \ref{lma-overlap-representative}, we get 
\begin{equation*}
    \sum_{C'\in \C^{(h)}} \ell^{h}|C'\cap C \cap \I_h| \leq \ell^{h}\frac{10}{\ell^h}\left(|C|+\sum_{C'\in \C^{(h)}} |C'\cap C| \right).
\end{equation*}
Finally note that each resource appears in at most $\ell$ configurations, hence
\begin{equation*}
    \sum_{j\leq h\leq k}\sum_{C'\in \C^{(h)}} |C'\cap C| \leq \ell |C|.
\end{equation*}
Putting everything together we conclude
\begin{align*}
    \sum_{j\leq h\leq k}\sum_{K\in \K^{(h)}} \ell^{h}|K\cap C \cap \I_h| &\leq \frac{1}{\ell} \sum_{j\leq h\leq k}\sum_{C'\in \C^{(h)}} \ell^{h}|C'\cap C \cap \I_h| + 1000\frac{d+\ell}{\ell}\log (\ell) |C| \\
    &\leq \frac{1}{\ell} \sum_{j\leq h\leq k}10\left( |C|+\sum_{C'\in \C^{(h)}}|C'\cap C|\right) + 1000\frac{d+\ell}{\ell}\log (\ell) |C|\\
    &\leq \frac{k-j}{\ell}10|C|+10|C|+1000\frac{d+\ell}{\ell}\log (\ell) |C|\\
    &\leq 20|C|+1000\frac{d+\ell}{\ell}\log (\ell) |C|\\
    &\leq 2000\frac{d+\ell}{\ell}\log (\ell) |C|.\qedhere
\end{align*}
\end{proof}

\begin{customcla}{\ref{cla:reconstruct_mu}}(restated) For any $C\in \K^{(\geq j)}$,
\begin{equation*}
    \frac{1}{\gamma^2}\sum_{K\in \K^{(j-1)}}|K\cap C \cap \I_{j-1}\cap O|\leq \mu \leq \frac{2}{\gamma} \sum_{K\in \K^{(j-1)}}|K\cap C \cap \I_{j-1}\cap O|.
\end{equation*}
\end{customcla}
\begin{proof}[Proof of Claim \ref{cla:reconstruct_mu}]
Note that we can write 
\begin{equation*}
    \mu = \sum_{i\in O\cap C} \frac{a_i+b_i-\gamma}{b_i} \leq \max_{i\in O\cap C}\left\lbrace \frac{a_i+b_i-\gamma}{a_ib_i} \right\rbrace \sum_{K\in \K^{(j-1)}}|K\cap C \cap \I_{j-1}\cap O|.
\end{equation*}
The reason for this is that each resource $i$ accounts for an expected loss of $(a_i+b_i-\gamma)/b_i$ while it is counted $a_i$ times in the sum 
\begin{equation*}
    \sum_{K\in \K^{(j-1)}}|K\cap C \cap \I_{j-1}\cap O|.
\end{equation*}
Similarly,
\begin{equation*}
    \mu = \sum_{i\in O\cap C} \frac{a_i+b_i-\gamma}{b_i} \geq \min_{i\in O\cap C}\left\lbrace \frac{a_i+b_i-\gamma}{a_ib_i} \right\rbrace \sum_{K\in \K^{(j-1)}}|K\cap C \cap \I_{j-1}\cap O|.
\end{equation*}
Note that by assumption we have that $a_i+b_i>\gamma$. This implies that either $a_i$ or $b_i$ is greater than $\gamma/2$. Assume w.l.o.g. that $a_i\geq \gamma/2$. Since by assumption $a_i\leq \gamma$ we have that 
\begin{equation*}
    \frac{a_i+b_i-\gamma}{a_ib_i}\leq \frac{b_i}{a_ib_i} =\frac{1}{a_i} \leq \frac{2}{\gamma}.
\end{equation*}
In the same manner, since $a_i+b_i>\gamma$ and that $a_i,b_i\leq \gamma$, we can write
\begin{equation*}
    \frac{a_i+b_i-\gamma}{a_ib_i}\geq \frac{1}{a_ib_i} \geq \frac{1}{\gamma^2}.
\end{equation*}
We therefore get the following bounds
\begin{equation*}
    \frac{1}{\gamma^2}\sum_{K\in \K^{(j-1)}}|K\cap C \cap \I_{j-1}\cap O|\leq \mu \leq \frac{2}{\gamma} \sum_{K\in \K^{(j-1)}}|K\cap C \cap \I_{j-1}\cap O|,
\end{equation*}
which is what we wanted to prove.
\end{proof}

\end{document}